\documentclass[runningheads]{llncs}

\usepackage[T1]{fontenc}
\usepackage{multirow}%
\usepackage{mathrsfs}%
\usepackage[title]{appendix}%
\usepackage{textcomp}% 
\usepackage{manyfoot}%
\usepackage{booktabs}%
\usepackage{algorithm}%
\usepackage{algpseudocode}%
\usepackage{listings}%
\usepackage{cite}
\usepackage{lineno}
\usepackage{array}

\usepackage{amsmath,amssymb,amsfonts}
\usepackage{graphicx}
\usepackage{textcomp}
\usepackage{xcolor}
\usepackage{bm}
\usepackage{stmaryrd}
\usepackage{xspace}
\usepackage{hyperref}

\begin{document}

\newcommand{\pr}{\operatorname{Pr}}
\newcommand{\sem}[1]{\ensuremath{\llbracket #1 \rrbracket}\xspace}
\newcommand{\sayan}[1]{\textcolor{blue}{#1}}
\title{Towards Unified Probabilistic Verification and Validation of Vision-Based Autonomy}

\titlerunning{Towards Unified Probabilistic Verification and Validation}

\author{Jordan Peper\inst{1}\thanks{The authors thank Alessandro Abate for a discussion of Bayesian validation of Markov models, and Thomas Waite and Radoslav Ivanov for sharing their implementation of the vision-based mountain car. \\This material is based on research sponsored by AFRL/RW under agreement number FA8651-24-1-0007 and the NSF CAREER award CNS-2440920. The U.S. Government is authorized to reproduce and distribute reprints for Governmental purposes notwithstanding any copyright notation thereon. The opinions, findings, views, or conclusions contained herein are those of the authors and should not be interpreted as representing the official policies or endorsements, either expressed or implied, of the DAF, AFRL, NSF, or the U.S. Government.}\and Yan Miao\inst{2} \and Sayan Mitra\inst{2}\and Ivan Ruchkin\inst{1}}

\institute{University of Florida, Gainesville, FL 32611, USA \and University of Illinois at Urbana-Champaign, Champaign, IL 61820, USA}

\maketitle
\vspace{-5mm}

\begin{abstract}
    Precise and comprehensive situational awareness is a critical capability of modern autonomous systems. Deep neural networks that perceive task-critical details from rich sensory signals have become ubiquitous; however, their black-box behavior and sensitivity to environmental uncertainty and distribution shifts make them challenging to verify formally. Abstraction-based verification techniques for vision-based autonomy produce safety guarantees contingent on rigid assumptions, such as bounded errors or known unique distributions. Such overly restrictive and inflexible assumptions limit the validity of the guarantees, especially in diverse and uncertain test-time environments. We propose a methodology that unifies the verification models of perception with their offline validation. Our methodology leverages interval MDPs and provides a flexible end-to-end guarantee that adapts directly to the out-of-distribution test-time conditions. We evaluate our methodology on a synthetic perception Markov chain with well-defined state estimation distributions and a mountain car benchmark. Our findings reveal that we can guarantee tight yet rigorous bounds on overall system safety.\\
    \noindent \textbf{Code: } {\scriptsize \href{https://github.com/Trustworthy-Engineered-Autonomy-Lab/unified-perception-vnv}{github.com/Trustworthy-Engineered-Autonomy-Lab/unified-perception-vnv} }\vspace{-4mm}
\end{abstract}

\section{Introduction}

\looseness=-1

Deep vision models are used to process rich sensory signals in safety-critical closed-loop systems -- from autonomous driving~\cite{bojarski_end_2016, pasareanu_closed-loop_2023} and robotic surgery~\cite{shvets_automatic_2018} to aircraft control~\cite{dong_deep_2021, dong_visual_2015}. %The proliferation of pre-trained deep models and user-friendly software -- such as the family of YOLO vision architectures~\cite{khanam_yolov11_2024} -- has transformed state-of-the-art visual inference pipelines into nearly turnkey technology.% and democratized their access across countless settings.
Nevertheless, verifying neural-perception systems remains a formidable challenge~\cite{mitra_formal_2025} due to their fragile nature~\cite{szegedy_intriguing_2014,kurakin_adversarial_2017}. The foremost obstacles include accurate modeling of visual signals in a verifiable way~\cite{pasareanu_closed-loop_2023} and devising logical specifications for the neural models~\cite{astorga_perception_2023, huang_safety_2017}. %. -- a critical step in model checking. %These difficulties are compounded by the high-dimensional and inherently black-box nature of deep visual inference architectures.%, which resists conventional verification methodologies.

% \sayan{
%A perception contract bounds the accuracy of a (vision-based) state estimator (or observer) such that these bounds are sufficient for formally verifying closed-loop invariants~\cite{tcad-HsiehLSJMM22} or Lyapunov functions~\cite{LiL4DC25}. %The accuracy bounds--which can be viewed as assumptions for the closed-loop guarantee--can be validated empirically or constructed via sampling, with certain probabilistic guarantees. 
% }

{\em Assume-Guarantee (A/G) reasoning} is a standard method for modular verification of complex systems~\cite{blundell_assume-guarantee_2005,sangiovanni-vincentelli_taming_2012}. In the A/G framework, one specifies the assumptions on the system's components and then verifies each component's guarantees under these assumptions, and this chain of reasoning establishes the system-level properties. 
To apply A/G reasoning to vision-based autonomous systems, a typical decomposition would be into  {\em observer} and {\em controller\/} components. The observer process would include the image formation (from the state) as well as state estimation (from one or more images).
% contained and does not   %} Since autonomous systems are difficult to verify due to hardware/software and hardware/environment interactions, contract-based design (CBD) rooted in A/G reasoning has been a natural solution to dealing with the inherent complexity of cyber-physical systems~\cite{}.

A/G reasoning has indeed been applied to vision-based systems with such decompositions. First, as a precursor to modeling visual uncertainties, A/G specifications have been extended to probabilistic uncertainties~\cite{nuzzo_stochastic_2019}. Another A/G-style technique,  called {\em perception contracts}~\cite{tcad-HsiehLSJMM22, astorga_perception_2023}, computes bounds on the accuracy of the observer such that these bounds are sufficient to verify the closed-loop system. Similarly, constructing stochastic assumptions on the observer has enabled the probabilistic verification of vision-based systems in~\cite{cleaveland_monotonic_2022,pasareanu_closed-loop_2023,calinescu_controller_2024,cleaveland_conservative_2025}. However, vision-based observers are well-known to be fragile and vulnerable to distribution shifts~\cite{wu_toward_2023,filos_can_2020}, which can violate the assumptions when the system is deployed in a new environment. Thus, a fundamental gap remains between A/G verification and the validity of those assumptions in the deployed system. 

This paper aims to extend the scope of verification guarantees for vision-based systems by measuring the validity of modeling assumptions. To this end, we introduce a \textit{unified verification \& validation methodology} to obtain a safety guarantee that is flexibly degraded when the assumptions -- in the form of a probabilistic model -- are likely to be invalid in a new environment. The key idea is that a probabilistic safety guarantee for the ground truth system $M_E$ modeled in environment $E$ is contingent on a statistical claim about the model's validity in a new environment $E'$, in the style of the Probably Approximately Correct (PAC) bounds~\cite{fu_probably_2014}. This will be represented with nested probabilistic assertions. %. statement, ``with confidence 

Our methodology consists of three steps: abstraction, verification, and validation. The abstraction step represents neural perception's uncertainty with confidence intervals, leading to an \textit{interval Markov decision process (IMDP)}  abstraction $\mathcal{M}_E$. This abstraction overapproximates the concrete system $M_E$ with confidence $\alpha$ (i.e., with probability $1-\alpha$, $\mathcal{M}_E$ will contain the behavior distribution of $M_E$). At a high level, $\mathcal{M}_E$ represents a set of possible closed-loop systems with uncertainties in observer and dynamics,  that are consistent with the data obtained from $M_E$. 

The verification step checks a system-level temporal property $\varphi$ on the constructed IMDP $\mathcal{M}_E$ with a probabilistic model checker~\cite{kwiatkowska_prism_2011,hensel_probabilistic_2020}. It produces an upper bound $\beta$ on the chance that a trajectory falsifies property $\varphi$. A combination of the first and second steps gives rise to a frequentist-style guarantee of the form: with confidence $\alpha$ in the dataset from which $\mathcal{M}_E$ was built, the chance that the underlying system $M_E$ produces a safe trajectory is at least $1-\beta$. 
 
The third step is to validate our IMDP abstraction $\mathcal{M}_E$ in a new environment $E'$. Here, we aim to measure the probability that the new concrete model $M_{E'}$ is contained in $\mathcal{M}_E$ -- and thus we pivot to the Bayesian perspective. Instead of formulating a frequentist hypothesis test (which would merely detect a significant difference between $E$ and $E'$), we construct a belief on the \textit{parameters} of $M_{E'}$ based on the new data. Then, ``intersecting" it with the probability intervals in IMDP $\mathcal{M}_E$ gives us a quantitative posterior $1-\gamma$ chance of the new environment $E'$ falling within the uncertainty of $\mathcal{M}_E$. Hence, with confidence $1-\gamma$, the system $M_{E'}$ satisfies the property $\varphi$ with probability $1-\beta$. %This yields our safety guarantee contingent on Bayesian validity: with probability $1-\gamma$, we find ourselves in the universe where the new concrete system $M_{E'}$ produces a trajectory that satisfies property $\varphi$ with probability $1-\beta$.
This nested guarantee elegantly combines the two approaches: frequentist and Bayesian.

We evaluate our methodology on two case studies: (1) a synthetic waypoint-following task and (2) a vision-based mountain car system. Results demonstrate that our methodology yields a flexible trade-off between safety within a model and validity of the model in the presence of perceptual uncertainty. Furthermore, our validation procedure reliably discriminates between systems operating within the training distribution and those operating in domain-shifted environments. %, influencing the overall safety guarantees in both offline and online deployment settings.

This paper makes three contributions: 
\begin{enumerate}
    \item A unified framework for verifying the safety of vision-based autonomous systems and validating them in a deployment environment.
    \item A flexible nested probabilistic guarantee for both verification and validation. 
    \item An evaluation of our framework in two case studies.  
\end{enumerate}

The rest of this paper is organized as follows. Section~\ref{sec:rw} surveys related work. % on probabilistic verification/validation, perception abstraction, and perception contracts.
Section~\ref{sec:pf} formulates our problem of making rigorous yet flexible guarantees that unify verification and validation. In Section~\ref{sec:app}, we describe our approach to addressing each subproblem in order, %followed by Section~\ref{sec:ovr}, where we wrap the verification and validation guarantees into a single overall end-to-end guarantee. 
Finally, we evaluate this methodology on two case studies in Section~\ref{sec:exp}, and then conclude with Section~\ref{sec:conc}.
% \vspace{-2mm}

\section{Related Work}
\label{sec:rw}

\looseness=-1
\paragraph{Contracts and assume/guarantee reasoning.}
Methodologies for assuring safety-critical software and systems originate from a rich
 tradition of compositionality~\cite{rushby_composing_2012,kwiatkowska_compositional_2013,bakirtzis_compositional_2021-1,ruchkin_confidence_2022}, contracts~\cite{graf_contract-based_2014,ruchkin_contract-based_2014,liebenwein_compositional_2020,incer_hypercontracts_2022}, and assume-guarantee reasoning~\cite{frehse_assume-guarantee_2004,kwiatkowska_assume-guarantee_2010,ruchkin_active:_2014,frenkel_assume_2022}. 
An engineer annotates each component with assumptions on its inputs and guarantees on its outputs; then, the two can be composed if the first component's guarantees satisfy the second's assumptions. 
%In this context, \textit{assumption synthesis} refers to deriving the minimal assumptions (e.g., on the environment) sufficient to provide the guarantees~\cite{cobleigh_learning_2003,elkader_automated_2015,mohammadinejad_mining_2020,anand_computing_2023}. 
 Recently, stochastic~\cite{li_stochastic_2017,nuzzo_stochastic_2019,nuzzo_electronic_2019} and neural-network~\cite{naik_robustness_2020,dreossi_compositional_2019} contracts attempt to handle uncertainties in learning components. Unfortunately, they do not offer a complete or straightforward solution to specifying and verifying vision-based systems.

\vspace{-2mm}
\paragraph{Probabilistic modeling and verification.} 
Probabilistic model checking~\cite{katoen_probabilistic_2016,kwiatkowska_probabilistic_2022} is a well-developed technique to check properties of various system models that exhibit probabilistic and non-deterministic transitions~\cite{howard_dynamic_2007,kwiatkowska_stochastic_2007}. The oft-cited scalability issue can be tackled with assume-guarantee reasoning~\cite{kwiatkowska_assume-guarantee_2010,feng_learning-based_2011,komuravelli_assume-guarantee_2012} and model counting~\cite{vazquez-chanlatte_model_2019,holtzen_model_2021}, as well as model reduction~\cite{bharadwaj_reduction_2017}, sampling~\cite{legay2014LSS,DArgenio2015SmartLSS}, and restructuring~\cite{ruchkin_integration_2019,camara_haiq_2020}. There is a vast literature about creating automata-based abstractions from detailed system descriptions~\cite{holtzen_probabilistic_2017,lomuscio_counter_2019,aichernig_probabilistic_2019,xie_mosaic_2023,mallik_compositional_2018}. The uncertainty from data can be accounted for by inferring confidence intervals over model parameters~\cite{cubuktepe_scenario-based_2020,badings_scenario-based_2022,alasmari_quantitative_2022}. This leads us to a class of models with sets of transition probabilities, known under the diverse names of uncertain/imprecise/interval/set/robust Markov models~\cite{puggelli_polynomial-time_2013,dinnocenzo_robust_2012,wolff_robust_2012,troffaes_model_2013,badings_efficient_2023,termine_robust_2021,zhao_bayesian_2024,jackson_formal_2021}. Our approach leverages these insights to construct and verify an interval Markov model of a vision-based closed-loop system. 

\vspace{-2mm}
\paragraph{Probabilistic model validation.} Whether a probabilistic model is valid can be quantified in several ways. If model validity is conceptualized as a stochastic binary event, it can be characterized by a probability estimate~\cite{corso_survey_2022} or confidence interval~\cite{fleiss_statistical_2003,bensalem_what_2024,gupta_distribution-free_2022}. Numerical notions of robustness~\cite{fainekos_robustness_2009,dong_reliability_2023,kwiatkowska_when_2023} and risk~\cite{majumdar_how_2020,chapman_risk-sensitive_2020} offer directional, smooth uncertainty quantification. At run time, model validity can be determined via confidence monitoring~\cite{ruchkin_compositional_2020} or parameter sampling~\cite{carpenter_modelguard_2021}. Our methodology adopts a distributional perspective~\cite{dutta_distributionally_2025,cauchois_robust_2024}: model validity is the inclusion of the true distribution into the modeled set, operationalized via the Bayesian approach to safety validation~\cite{moss_bayesian_2023}.

\vspace{-2mm}
\paragraph{Perception abstractions and contracts.} 
To model vision-based perception, most formal approaches use simple error probabilities and noise parameters~\cite{badithela_evaluation_2022,cleaveland_monotonic_2022,dreossi_verifai_2019,xie_mosaic_2023,wang_bounding_2021}. Other approaches build detailed models of specific vision systems, such as aircraft landing cameras~\cite{santa_cruz_nnlander-verif_2022} or racing car LiDARs~\cite{ivanov_case_2020}. Sadly, such approaches have limited expressiveness and do not generalize to real-world systems. To obtain guarantees, perception contracts~\cite{tcad-HsiehLSJMM22,LiL4DC25, astorga_perception_2023,sun_learning-based_2024} assume a bound on the accuracy of a vision-based state estimator such that these bounds are sufficient for formally verifying closed-loop invariants~\cite{tcad-HsiehLSJMM22} or Lyapunov functions~\cite{LiL4DC25}.
In a similar vein, recent works on model checking vision systems have constructed a verifiable probabilistic model~\cite{cleaveland_monotonic_2022,pasareanu_closed-loop_2023,cleaveland_conservative_2023,cleaveland_conservative_2025}. None of these approaches have been investigated (let alone provide guarantees) when the visual distribution changes from the one for which the model was built. 

% The accuracy bounds--which can be viewed as assumptions for the closed-loop guarantee--can be validated empirically or constructed via sampling, with certain probabilistic guarantees. 
% These methods have been used to analyze visual lane keeping and auto-landing systems and identify their safe operating domains (ODDs), but it is unclear what happens when a system is deployed outside its ODD.

\vspace{-2mm}
\section{Problem Formulation}
\label{sec:pf}

In this paper, we address the \emph{composite problem} of (i) building a sound abstraction of perception from sample executions of a vision-based autonomous system, (ii) providing rigorous safety guarantees, and (iii) validating these guarantees in an offline deployment setting. We first introduce our system notation, then transition into our \emph{three primary subproblems}.

% (i) providing rigorous safety guarantees for a vision-based autonomous system, and (ii) validating these guarantees in a deployment setting. 

Consider the closed-loop discrete-time dynamical system $M_E$ that is composed of black-box sensing and perception processes (hereafter, ``state estimation") operating under the randomness of the latent environment distribution $E$, and subsequently a deterministic, known control loop:
\begin{equation}
\label{eqn:full_model}
M_E:\quad
\left\{
    \begin{aligned}
        &\left.\begin{aligned}
            % &\text{Latent environment:} && e \sim \pr_E(e) \\
            &\text{Sensing:} && o_t = g(s_t, e_t \sim E) \\
            &\text{Perception:} && \hat{s}_t = h(o_t)
        \end{aligned}\right\}\text{ State estimation}\\[10pt]
        &\left.\begin{aligned}
            &\text{Control policy:} && u_t = \pi (\hat{s}_t) \\
            &\text{Dynamics:} && s_{t+1} = f( s_t, u_t)
        \end{aligned}\right\}\text{ Control loop}
    \end{aligned}
\right.
\end{equation} 

The operational domain of the closed-loop system $M_E$ is deployed in an environment $E$ -- an unknown probability distribution that characterizes the uncertain behavior of a domain. At any time $t$, a single realization $e_t \in \mathcal{E}$ is drawn from the probability distribution $E$. Further, $s_t \in S \subset \mathbb{R}^n$ is the system state (e.g., position, class label); $o_t \in O \subset \mathbb{R}^k$ is the observation of this state and environment (e.g., a camera image or a LiDAR scan); $\hat{s}_t \in \hat{S} \subset \mathbb{R}^n$ is an estimate of the system state; and $u_t \in U \subset \mathbb{R}^m$ is the control action. An observation of the latent environment at a given state is made with the sensor $g : S \times \mathcal{E} \rightarrow O$, and the state estimates are made with the estimator $h : O \rightarrow  \hat{S}$. Once $\hat{s}_t$ is computed from the state estimation process, a control action is determined with the policy $\pi : \hat{S} \rightarrow U$, and the state is updated with the dynamics $f : S \times U \rightarrow S$.

Given system $M_E$ in some environment $E$ (e.g., an autonomous drone or self-driving car), its \emph{execution} is a pair of true and estimated trajectories $\bm{\tau} =\left(\bm{s}, \boldsymbol{\hat{s}}\right) \sim \sem{M_E}$, where operator $\sem{.}$ means the model semantics (in this case, the distribution over trajectories). Each trajectory is based on the trace of \emph{environment realizations} $\boldsymbol{e}$.  In the case of a model $M$ with a fixed environment $E$, $\sem{M_E}$ can be interpreted as the distribution over all possible $\bm{\tau}$ executions of $M_E$. 

\looseness=-1
Furthermore, we will construct an abstraction function $\psi$ such that $\psi(\bm{\tau}) = \bm{\tau}^{\psi}$, where $\bm{\tau}^{\psi}$ is a discrete-space execution of $M_E$. This induces an abstract model $\mathcal{M}_E$. Here, the semantics $\sem{\mathcal{M}_E}$ is a set of distributions over $\bm{\tau}$ resulting from conservative over-approximation. We slightly abuse the notation  $ \sem{\psi\left(M_E\right)} $ to mean the distribution over abstract, discrete-space executions of $M_E$.

With this notation, we describe the three problems solved in this work:

% Subproblem 1
\begin{problem} [Data-driven abstraction] 
\label{prb:abs-prob}
    Suppose we have a real physical system $M_E$, with a known controller $\pi$ and dynamics $f$, from which we collect a dataset of \emph{executions} $\mathcal{D}^{train}$ drawn according to $\sem{M_E}$. Our goal is to use this dataset to derive an abstract model $\mathcal{M}_E$ that satisfies:
    \begin{equation}
        \sem{\psi(M_E)} \in \sem{\mathcal{M}_E} \text{,}
    \end{equation}
\end{problem}

% Subproblem 2
\begin{problem} [Safety verification]
\label{prb:ver-prob}
    Suppose we are given $\mathcal{M}_E$, which is an abstract model of $M_E$, and a co-safe linear temporal logic (LTL) property $\varphi$. The problem is to compute the worst-case logical satisfaction probability of the LTL predicate $\varphi$ on $\mathcal{M}_E$:
    \begin{equation}
    \label{prb:ver}
        \Pr\!^{min}_{ \sem{ \mathcal{M}_E}} \left( \varphi  \right) \ge 1-\beta
    \end{equation}
\end{problem}

% Subproblem 3
\begin{problem} [Design-time validation]
\label{prb:off-val-prob}
    Suppose we are given $\mathcal{M}_E$, which is an abstract model of $M_E$, a novel deployment environment $E'$, a dataset $\mathcal{D}^{val}$ of \emph{executions} $\tau$, where $\tau \sim \sem{M_{E'}}$, and a prior belief $\mathbb{M}^{pr}$ about the model $M_{E'}$ which is a distribution over candidate models $M$. The problem is to obtain the posterior $\mathbb{M}^{po}$ by updating the belief $\mathbb{M}^{pr}$ with $\mathcal{D}^{val}$, and then compute the conformance confidence $\gamma$ that $\sem{M_{E'}}$ is characterized by some distribution in $\mathcal{M}$. Formally:
    \begin{equation}
        \Pr\!_{\mathbb{M}^{po}} \left(   \sem{M} \in \sem{\mathcal{M}_E} \mid M \sim \mathbb{M}^{po} \right) \ge 1-\gamma
    \end{equation}

\end{problem}

\section{Approach: Abstraction, Verification, and Validation}
\label{sec:app}

This section describes our three-step methodology illustrated in Figure~\ref{fig:megaflow}: %we address the three primary problems put forth in this paper:
\begin{itemize}
    \item Subsection~\ref{sub:abstraction} expands on Problem~\ref{prb:abs-prob} of forming a sound abstraction of a vision-based autonomous system and proposes leveraging an \emph{interval Markov decision process} (IMDP) to model imprecise knowledge about the true distribution of state estimates.
    \item Subsection~\ref{sub:verification} addresses Problem~\ref{prb:ver-prob} of verifying the system abstraction by leveraging well-explored probabilistic model checking techniques.
    \item Subsection~\ref{sub:off-validation} reiterates Problem~\ref{prb:off-val-prob} of probabilistically validating the abstraction under a novel deployment environment and proposes a technique for modeling a posterior distribution over model parameters to measure their similarity to the abstraction parameters.
\end{itemize}

% SOURCE: https://docs.google.com/presentation/d/1t0UVOPtLDRONy47Daa6uSPZhtSnizRCA9ZqN9pmf8wU/edit#slide=id.g2d808f81989_0_0
\begin{figure}[htb]
\vspace{-6mm}
\centering
\includegraphics[width=\textwidth]{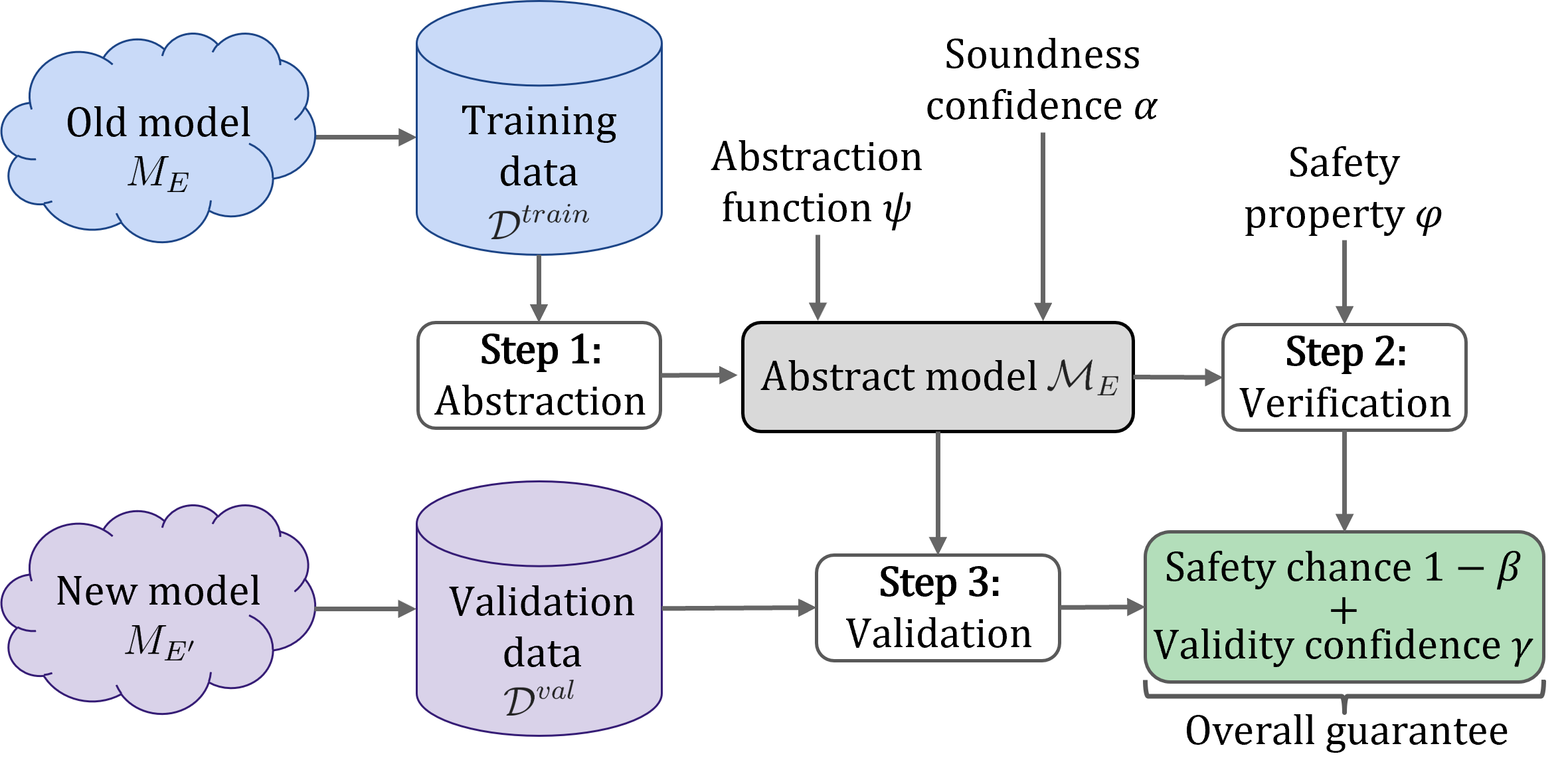}
\vspace{-8mm}
\caption{Our verification \& validation methodology in three steps.}
\vspace{-10mm}
\label{fig:megaflow}
\end{figure}

\paragraph{Running example.} We utilize the OpenAI Gym~\cite{brockman_openai_2016} MountainCar-v0 benchmark (a widely adopted testbed for reinforcement learning) as a running example to clarify the algorithms throughout this section. This environment models a two-dimensional control problem with state variables $s=(x, v) \in [-1.2, 0.6] \times [-0.07,0.07]$, where $x$ is the horizontal position of the vehicle and $v$ is its velocity. At each timestep, the agent selects from actions $u \in \{0, 1, 2\}$, corresponding respectively to the application of a negative force (throttle leftward), no force, or a positive force (throttle rightward).

The objective of the Mountain Car is to reach the top of a steep incline by gaining momentum between the hills. The system is equipped with a pre-trained vision module that predicts the $x$-position of the car from third-person scene images. These predictions, alongside the ground-truth velocity $v$, are fed into a pre-trained Deep Q-Network (DQN) control policy $\pi$ to obtain control action $u$.

\subsection{Constructing abstraction from data}
\label{sub:abstraction}

\looseness=-1
Recall the real-world closed-loop, discrete-time dynamical system $M_E$ introduced in Equation~\ref{eqn:full_model}, which consists of a stochastic perception pipeline $h \circ g$ for state estimation --- where $g$ represents the observation process and $h$ denotes the perception or inference model --- governed by the randomness induced by a latent environment distribution $E$. This is followed by a known, deterministic control policy and dynamics composition $f \circ \pi$, where $\pi$ is the policy mapping state estimates to control actions, and $f$ propagates the system state forward in time.

An execution $\tau$ of $M_E$ is drawn by sampling $e$ from the operational environment $E$, making an observation of this environment $g : S \times \mathcal{E} \rightarrow O$, estimating the manifest state variables $h : O \rightarrow  \hat{S}$, computing a control action $\pi : \hat{S} \rightarrow U$, and propagating the state $f : S \times U \rightarrow S$. Fundamentally, $\tau$ is a random variable which follows the unknown distribution of trajectories $\sem{M_E}$ induced by $E$. This formulation enables the precise definition of a sound abstraction in Definition~\ref{def:sound-abstraction}.

\begin{definition}[Sound abstraction]
\label{def:sound-abstraction}
    Given the system $M_E$ composed of deterministic components $g$, $h$, $\pi$, and $f$, subject to stochasticity from the latent environment distribution $E$, and an abstraction function $\psi: \mathbb{R}^n \rightarrow \mathbb{Z}^n$, a \emph{sound abstraction} $\mathcal{M}_E$ is an abstract model that satisfies the inclusion:
    $$\sem{\psi(M_E)} \in \sem{\mathcal{M}_E}$$
\end{definition}

We aim to construct a \emph{sound abstraction} of $M_E$ that over-approximates its behavior, leading to conservative safety estimates in the next subsection. The abstraction problem is intentionally formulated in a distribution-agnostic manner, i.e., the parameters of $\sem{M_E}$ are unknown, so we must construct $\mathcal{M}_E$ from a dataset $\mathcal{D}^{train}$ of $N$ trajectories $\tau$. The \emph{data-driven abstraction problem} becomes constructing a \emph{statistically sound abstraction} per Definition~\ref{def:stat-sound-absteaction} since $\mathcal{D}^{train}$ is a finite set of executions.

\begin{definition}[Statistically sound abstraction]
\label{def:stat-sound-absteaction}
    Given the system $M_E$ composed of deterministic components $g$, $h$, $\pi$, and $f$, subject to stochasticity from the latent environment distribution $E$, and an abstraction function $\psi: \mathbb{R}^n \rightarrow \mathbb{Z}^n$, a dataset $\mathcal{D}^{train}$ of trajectories $\tau$ sampled i.i.d from $\sem{M_E}$, and a confidence level $\alpha$, a \emph{statistically sound abstraction} $\mathcal{M}_E$ is an abstract model that satisfies:
    % $$\Pr \left ( \exists \sem{M} \in \sem{\mathcal{M_E}} : \bm{\tau}_1, \bm{\tau}_2, \dots, \bm{\tau}_N \sim \sem{M} \right ) \ge 1-\alpha$$
    $$\Pr\!_{\mathcal{D}^{train}} \left ( \sem{\psi(M_E)} \in \sem{\mathcal{M}_E} \right ) \ge 1-\alpha$$
\end{definition}

We first demonstrate how the real system $M_E$ can be soundly transformed into a probabilistic automaton, then reiterate that its unknown transition probabilities must be statistically inferred from data.

Our approach considers state estimation and state propagation within $M_E$ to be separable, yet sequentially-composed processes. In the concrete spaces $S$ and $\hat{S}$, $g \circ h :  S \times \mathcal{E} \times \hat{S} \rightarrow [0, 1]$ is a probabilistic mapping, and $f \circ \pi : \hat{S} \rightarrow S$ is an injective mapping. This behavior can be represented as a simple Markov chain, as in Definition~\ref{def:mc} below, where the state space $X = S \times \hat{S}$ consists of pairs of states and estimates, and $T : (S \times \hat{S}) \times (S \times \hat{S}) \rightarrow [0, 1]$ is induced by the latent environment distribution $E$:

\begin{definition} [Markov Chain]
\label{def:mc}
    A \emph{Markov chain} is a triple $(X, x_0, T)$, where $X$ is the state space, $x_0 \in X$ is the initial state, and $T : X \times X \rightarrow [0, 1]$ is the probabilistic transition function.
\end{definition}

However, %explicitly instantiating this theoretical Markov chain is an intractable task due to the continuous nature of the state space. 
model checking continuous-state Markov processes is challenging with current tools. One common workaround is to discretize the state space of the Markov chain, yielding a countable set of states and transitions. Let $\psi : \mathbb{R}^n \rightarrow \mathbb{N}^n$ be an abstraction function that maps the otherwise continuous space $S$ into a countably infinite space $S^{\psi}$ (and $\hat{S} \rightarrow \hat{S}^{\psi}$) of $n$-dimensional hypercubes.

In the abstract spaces $S^{\psi}$ and $\hat{S}^{\psi}$, $g \circ h :  S^{\psi} \times \mathcal{E} \times \hat{S}^{\psi} \rightarrow [0, 1]$ is still a probabilistic mapping, but $f \circ \pi : \hat{S} \twoheadrightarrow S$ becomes a surjective mapping to over-approximate the model behavior. From this, we construct a discrete state Markov decision process, seen in Definition~\ref{def:mdp} below, where $X=S^{\psi} \times \hat{S}^{\psi}$.

\begin{definition} [Markov Decision Process]
\label{def:mdp}
    A \emph{Markov decision process} is a tuple $(X, x_0, \omega, \delta, T)$, where $X$ is the state space, $x_0 \in X$ is the initial state, $\omega$ is an alphabet of action labels, $\delta: X \rightarrow \omega$ is an adversary that resolves non-determinism, and $T : X \times \omega \times  X \rightarrow [0, 1]$ is the probabilistic transition function.
\end{definition}

Recall that we do not actually know the exact probabilistic transition function $T$. However, the information about it is present in our dataset $\mathcal{D}^{train}$. To perform a statistically rigorous approximation of this ground-truth MDP, which is a conservative discrete-state abstraction of the ground-truth MC, we use confidence intervals (CI) in place of single transition probabilities, resulting in an interval Markov decision process outlined in Definition~\ref{def:imdp}, still with $X=S^{\psi} \times \hat{S}^{\psi}$

\begin{definition} [Interval Markov Decision Process]
\label{def:imdp}
    An \emph{interval Markov decision process} is a tuple $(X, x_0, \omega, \delta, \Delta)$, where $X$ is the state space, $x_0 \in X$ is the initial state, $\omega$ is an alphabet of action labels, $\delta: X \rightarrow \omega$ is an adversary that resolves non-determinism, and $\Delta: X \times \delta \times X \rightarrow \left\{   \left[a, b \right] \mid 0 \le a \le b \le 1  \right\}$ is the interval probability transition function.
\end{definition}

The interval probability function $\Delta$ is computed with Algorithm~\ref{alg:conf} such that with $1-\alpha$ confidence, the intervals in $\Delta$ bracket the true transition probabilities in $T$, thereby capturing our uncertain knowledge about the true distribution induced by $E$. In practice, we chose to use the Clopper-Pearson interval, which is the maximally conservative confidence interval for binomial processes. It is referred to as the ``exact'' interval since it offers precise $1-\alpha$ coverage of the binomial density, unlike other CIs that rely on asymptotic assumptions like the central limit theorem~\cite{newcombe_two-sided_1998}.

The discovered interval probability function $\Delta = \texttt{ConfInt}\left(\mathcal{D}^{train}, \psi, \alpha \right)$ determines the parameters of the desired abstraction. In practice, the interval probabilities correspond to a particular abstract estimate $\hat{s}^{\psi}$ at an abstract state $s^{\psi}$. However, we still require an underlying structure that organizes the allowed transitions from one state to another, modeling sequential state estimation  $h \circ g$ and conservatively modeling state propagation $f \circ \pi$. We utilize Algorithm~\ref{alg:sm} to abstract $M_E$ into a codified automaton parameterized by $\Delta$.

% \vspace{-4mm}
\begin{algorithm}[H]
\caption{\texttt{ConfInt}: compute confidence intervals for IMDP $\mathcal{M}_E$}
\begin{algorithmic}[1]
\Require Dataset $\mathcal{D}^{train}$, abstraction function $\psi$, confidence level $\alpha$
\Ensure confidence intervals $\Delta$ indexed by abstract state $s^{\psi}$ and abstract estimate $\hat{s}^{\psi}$
\State Initialize $\Delta \gets \{\}$
\State Let $N \gets |\hat{S}^\psi|$ \label{line:conf0}
% \State Concatenate the trajectories: $\bm{\tau}_{\text{all}} \gets \bigcup_{i=1}^N \bm{\tau}_i$
\State $\bm{s}^{\psi} \gets [\psi(\bm{s}_i) \mid (\bm{s}_i, \bm{\hat{s}}_i) \in \mathcal{D}^{train}]$ \Comment{Bin the state data} \label{line:conf1}
\State $\bm{\hat{s}^{\psi}} \gets [\psi(\bm{\hat{s}}_i) \mid (\bm{s}_i, \bm{\hat{s}}_i) \in \mathcal{D}^{train}]$  \label{line:conf2}
\ForAll{unique $s^{\psi} \in \bm{s}^{\psi}$} \label{line:conf3}

    \State Let $I \gets \{ i \mid \bm{s}^{\psi}[i] = s^{\psi} \}$ \Comment{Vector indices where $s^{\psi}$ exists}
    \State Let $\hat{S}_{s^{\psi}} \gets \{ \bm{\hat{s}^{\psi}}[i] \mid i \in I_{s^{\psi}} \}$ \Comment{Subset of $\hat{S}^{\psi}$ indexed by $I$} \label{line:conf4}
    % \State $\bm{\hat{s}^{\psi}} \gets [\psi(\bm{\hat{s}}_i) \mid (\bm{s}_i, \bm{\hat{s}}_i) \in \mathcal{D}^{train}]$
    
    \ForAll{unique $\hat{s}^{\psi} \in \hat{S}^{\psi}_{s^{\psi}}$} \label{line:conf5}
        \State $n \gets \#\{\hat{s}^{\psi} \text{ in } \hat{S}^{\psi}_{s^{\psi}}\}$
        \State $CI \gets [\operatorname{Beta}^{-1}( \alpha/2N, n , |\hat{S}^{\psi}_{s^{\psi}}| - n + 1), \operatorname{Beta}^{-1}(1-\alpha/2N, n+1, |\hat{S}^{\psi}_{s^{\psi}}|-n)]$
        \State $\Delta[s^{\psi}, \hat{s}^{\psi}] \gets CI$ \label{line:conf6}
    \EndFor 
\EndFor
\State \Return $\Delta$
\end{algorithmic}
\label{alg:conf}
\end{algorithm}

\vspace{-12mm}

\begin{algorithm}[H]
\caption{\texttt{DynStruct}: build the transition structure of IMDP $\mathcal{M}_E$}
\begin{algorithmic}[1]
\Require Control loop $f \circ \pi$, bounded state and estimate spaces $\mathcal{S}$ and $\hat{\mathcal{S}}$, abstraction function $\psi$
\Ensure $\texttt{IMDP}$ mapping abstract states to possible successor abstract states

\State Initialize $\texttt{IMDP} \gets \{\}$

\State Define $\mathcal{S}^{\psi} \gets \psi(S)$ \label{line:sm0}
\State Define $\hat{\mathcal{S}}^{\psi} \gets \psi(\hat{S}) $ \label{line:sm1}

\ForAll{$(s^{\psi}, \hat{s}^{\psi}) \in \mathcal{S}^{\psi} \times \hat{\mathcal{S}}^{\psi}$}

    \State $s_{\text{lb}} \gets \min [\psi^{-1}(s^{\psi})]$, $s_{\text{ub}} \gets \max [\psi^{-1}(s^{\psi})]$ \label{line:sm2}
    \State $\hat{s}_{\text{lb}} \gets \min [\psi^{-1}(\hat{s}^{\psi})]$, $\hat{s}_{\text{ub}} \gets \max [\psi^{-1}(\hat{s}^{\psi})]$ \label{line:sm3}
    
    \State $\Delta s_{\min} \gets \infty$, $\Delta s_{\max} \gets -\infty$ \label{line:sm4}
    \ForAll{$s \in \{s_{\text{lb}}, s_{\text{ub}}\}$, $\hat{s} \in \{\hat{s}_{\text{lb}}, \hat{s}_{\text{ub}}\}$}
        \State $\delta s \gets s - f(s, \pi(\hat{s}))$
        \State Update $\Delta s_{\min}$ and $\Delta s_{\max}$ elementwise
    \EndFor \label{line:sm5}
    \State $s_{\text{next, min}} \gets s_{\text{lb}} + \Delta s_{\min}$
    \State $s_{\text{next, max}} \gets s_{\text{ub}} + \Delta s_{\max}$
    \State $s^{\psi}_{\text{next}} \gets \{  \psi(s_i)  \mid s_i \in [s_{\text{next, min}}, s_{\text{next, max}}] \}$ \label{line:sm6}
    \State $\texttt{current\_state} \gets (s^{\psi}, \hat{s}^{\psi})$
    \State $\texttt{possible\_states} \gets \emptyset$
    
    \ForAll{$s^{\psi} \in s^{\psi}_{\text{next}}$}
        \State Add $s^{\psi}$ to $\texttt{possible\_states}$
    \EndFor
    \State $\texttt{IMDP}[\texttt{current\_state}] \gets \texttt{possible\_states}$
\EndFor
\State \Return $\texttt{IMDP}$
\end{algorithmic}
\label{alg:sm}
\end{algorithm}

\vspace{-10mm}

% \noindent
% \textbf{Example.} Consider a simple vision-based autonomous car tasked with driving through a fixed waypoint. Let $s = \langle p_{ego}, p_{way} \rangle \in \mathbb{R}^2$ be the state vector, where $p_{ego}$ is the 1D position of the car, and $p_{way}$ is the 1D position of the waypoint. During closed-loop execution, the agent observes its surroundings by taking an image from a front-mounted camera, then estimates its state $\hat{s}$, which informs the control action $u$, and subsequently propagates the agent to its next state. Our objective is to construct an IMDP that overapproximates the behavior of the real closed-loop system using the proposed method above, and explain Algorithms~\ref{alg:conf} and~\ref{alg:sm} along the way.

% \vspace{-4mm}

\paragraph{Running example.} Recall the mountain car state space $s=(x,v) \in [-1.2,0.6]\times[-0.07,0.07]$. We select an abstraction function $\psi$ partitioning the $x$-space by intervals of width $0.05$ and the $v$-space by intervals of width $0.005$, yielding $S^{\psi}=\{-24,\dots,12\}\times\{-14,\dots,14\}$. Similarly, we abstract estimation errors $x - \hat{x} \in [-0.5,0.5]$ into intervals of width $0.1$, obtaining $\hat{S}^{\psi}=\{-5,\dots,5\}$. For instance, a concrete state $s_k=(-0.3,0.06)$ abstracts to $\psi(s_k)=(-6,12)$ and concretizes to the subspace $\psi^{-1}(\psi(s_k)) = [-0.30,-0.25]\times[0.06,0.065]$. The estimation abstraction follows similarly.

Algorithm~\ref{alg:conf} computes intervals $\Delta$ indexed by abstract state-estimate pairs $(s^\psi, \hat{s}^\psi)$ by abstracting samples in $\mathcal{D}^{train}$ (Lines~\ref{line:conf1}--\ref{line:conf2}), counting occurrences per abstract state (Lines~\ref{line:conf3}--\ref{line:conf4}), then applying Beta quantiles to estimate confidence intervals (Lines~\ref{line:conf5}-\ref{line:conf6}). For example, given samples $(0,0,-1),(0,0,0),(0,0,0),$ $(0,0,1)$ at $s^\psi=(0,0)$, relative frequencies are computed by simply dividing their relative occurrence $n_i$ by $N = |\hat{S}^\psi|=11$ ($0.09$, $0.18$, and $0.09$, respectively) and bracketed by confidence intervals discovered with the Beta quantile function.

Algorithm~\ref{alg:sm} constructs an IMDP by concretizing abstract states and estimates (Lines~\ref{line:sm0}--\ref{line:sm3}), iterating over every abstract state-estimate pair in $S^\psi \times \hat{S}^\psi$, such as $(-6,12,0)$, to define concrete subspaces (e.g., $[-0.30,-0.25]\times[0.060,0.065]$ $\times[0.0,0.1]$). The algorithm loops through the corners of these subspaces, computes control actions from estimates using policy $\pi$, and propagates the concrete states through the dynamics $f$ (Lines~\ref{line:sm4}--\ref{line:sm5}). For instance, the corner $(-0.3,0.065,0.0)$ yields a concrete future state by evaluating $u=\pi(0.0)$ and $(x,v)=(-0.3,0.065)$ through $f$. The reachable subspace in one step is then abstracted into discrete successor states (Line~\ref{line:sm6}), updating the IMDP structure accordingly.

\vspace{-2mm}

\subsection{Verifying the safety of abstraction}
\label{sub:verification}

We now address the problem of formally verifying the safety of a vision-based autonomous system. We seek to verify that the model $\mathcal{M}_E$, which is an $\alpha$-sound abstraction of $M_E$ per Definition~\ref{def:stat-sound-absteaction}, satisfies a given co-safe linear temporal logic (LTL) property (a finite-time-horizon liveness property) $\varphi$ with at least $1-\beta$ probability. Formally, we aim to compute the logical satisfaction probability of the LTL predicate $\varphi$ on $\mathcal{M}_E$. The property $\varphi$ is a Boolean predicate that represents a real-world constraint on the temporal behavior of the closed-loop system $M_E$. The predicate is written in LTL, translated from a natural language safety requirement such as ``the car must make an emergency stop within 5 seconds whenever it observes a pedestrian in the road" to $G \big(\text{pedestrian} \implies F_{[0, 5]}$ $[\text{speed}=0] \big)$.

Since the abstraction $\mathcal{M}_E$ of the real system $M_E$ is stochastic due to the latent environment distribution $E$, we are interested in verifying that it is \emph{probabilistically safe} with at least $1-\beta$ probability per Definition~\ref{def:prob-safe}:
\begin{definition} [Probabilistic safety]
\label{def:prob-safe}
    An abstract model $\mathcal{M}_E$ \emph{satisfies} an LTL safety property $\varphi$, denoted $\mathcal{M}_E \vDash \varphi$, with at least $1-\beta$ probability iff:
    $$\Pr\!_{\tau \sim \sem{\mathcal{M}_E}}\left( \tau \vDash \varphi \right) \ge 1-\beta,$$
    where $\tau$ is a random trajectory sampled from $\sem{\mathcal{M}_E}$.
\end{definition}

We approach this problem by leveraging an existing probabilistic model checker on the output of Algorithm~\ref{alg:sm}. Since our derived abstraction $\mathcal{M}_E$ is an interval Markov decision process, tools like $\texttt{PRISM}$~\cite{kwiatkowska_prism_2011} and $\texttt{STORM}$~\cite{hensel_probabilistic_2020} are capable of verifying PCTL properties on IMDPs. These model checkers can evaluate the minimum probability that any execution of the IMDP $\bm{\tau}$ is satisfied on the property LTL $\varphi$. For standard MDP verification, the model checker resolves the nondeterminism by identifying the worst-case scheduler $\delta$; however, for IMDPs, the procedure further requires determining the worst-case probability distributions encoded in the probability interval function $\Delta$.

Since $\mathcal{M}_E$ is a statistically sound abstraction of $M_E$ with at least $1-\alpha$ confidence, and $\mathcal{M}_E$ is probabilistically safe on the property $\varphi$ with at least $1-\beta$ chance, we assert a nested guarantee that ``the trajectories are likely safe assuming that the abstraction is sound'' for the real system $M_E$ in Theorem~\ref{thm:id-safety} below (see Appendix~\ref{app:proof-id} for the proof):

\begin{theorem} [In-distribution soundness and safety]
\label{thm:id-safety}
Let $M_E$ be a concrete system with the state/estimation spaces $S$ and $\hat{S}$, composed of deterministic components $g$, $h$, $\pi$, and $f$, subject to stochasticity from the latent environment distribution $E$, from which originates training data $\mathcal{D}^{train}$ consisting of random trajectories $\tau$ drawn i.i.d. from $\sem{M_E}$; let $\psi$ be an abstraction function; let $\mathcal{M}_E$ be an \emph{$\texttt{IMDP} = \texttt{DynStruct}(\pi, f, S, \hat{S}, \psi)$} parameterized by \emph{$\Delta = \texttt{ConfInt}\left(\mathcal{D}^{train}, \psi, \alpha \right)$}, where $\alpha$ is a soundness confidence level.
    
    If it holds that:

    \vspace{-3mm}
    \begin{itemize}
        \item $\mathcal{M}_E$ is a statistically $\alpha$-sound abstraction of $M_E$ per Def.~\ref{def:stat-sound-absteaction}. 
        \item $\mathcal{M}_E$ is safe for LTL safety property $\varphi$ with probability $1-\beta$ per Def.~\ref{def:prob-safe}. 
    \end{itemize}
    \vspace{-3mm}

    Then:
    $$
    \Pr\!_{\mathcal{D}^{train}} \left[\Pr\!_{\tau \sim \sem{M_E}} \left( \tau \vDash \varphi \right) \ge 1-\beta \right] \ge 1-\alpha
    $$
\end{theorem}

\subsection{Validating abstraction on new data}
\label{sub:off-validation}

Although verification yields rigorous guarantees about the system $M_E$, what if the system is deployed to an environment $E'$ that may differ from the original environment $E$? Since the original safety guarantees were derived using the IMDP abstraction $\mathcal{M}_E$, which is parameterized by the set of probability intervals $\Delta = \texttt{ConfInt}\left(\mathcal{D}^{train}, \psi, \alpha \right)$, these guarantees no longer hold meaningful value without reconciling the new domain $E'$ with the old domain $E$. Verifying that the abstract model $\mathcal{M}_E$ still describes the behaviors of the real system $M$ under this novel environment $E'$ is critical to asserting a meaningful end-to-end guarantee free of run-time domain assumptions. 

This section tackles an \emph{offline validation setting} where we collect a novel dataset $\mathcal{D}^{val}$, consisting of $L$ trajectories drawn independently from the new environment's trajectory distribution $\sem{M_{E'}}$. The objective is to infer the posterior distribution $\mathbb{M}^{po}$ of the parameters of $M_{E'}$ by updating a prior distribution $\mathbb{M}^{pr}$ with data, and then compute the confidence that $M_{E'}$ is characterized by some distribution in $\sem{\mathcal{M}}$. Formally, we aim to check whether $\mathcal{M}_E$ is a \emph{statistically valid abstraction} of $M_{E'}$:

\begin{definition} [Statistically valid abstraction]
\label{def:prob-valid}
    Given an abstract model $\mathcal{M}_E$, a validation dataset $\mathcal{D}^{val}$ sampled i.i.d. from $M_{E'}$, and a prior belief $\mathbb{M}^{pr}$ about the parameters of $M_{E'}$, $\mathcal{M}_E$ is a \emph{statistically valid abstraction} of $M_{E'}$ with $\gamma$ confidence if:
    $$\Pr\!_{\mathbb{M}^{po}} \left(   \sem{M} \in \sem{\mathcal{M}_E} \mid M \sim \mathbb{M}^{po} \right) \ge 1-\gamma$$
    where $\mathbb{M}^{po}$ is a posterior belief about the parameters of $M_{E'}$.
\end{definition}

\paragraph{Bayesian validation.} The objective of model validation is to measure the \emph{extent} to which the behaviors of $M_{E'}$ are characterized by the abstract model $\mathcal{M}_E$ through statistical inference. We begin by formulating our null hypothesis:
\begin{equation}
    H_0 := \sem{M_E} \in \sem{\mathcal{M}_E},
\end{equation}
asserting that the distribution of trajectories within the true system is contained within the abstract model's set of trajectory distributions. We then acquire a validation dataset $\mathcal{D}^{val}$.

Two fundamental paradigms exist for testing $H_0$: frequentist and Bayesian. In the frequentist paradigm, $\mathcal{D}^{val}$ is viewed as random data generated under a fixed hypothesis (either true or false). Hypothesis testing involves computing the probability $\Pr(\mathcal{D}^{val} \mid H_0)$ of observing the validation dataset under the assumption $H_0$ is true. This probability is then used to derive a p-value, resulting in a \emph{binary} outcome (reject or fail to reject) based on its comparison to a pre-specified significance threshold, which controls the Type I error rate.

In contrast, the Bayesian paradigm treats $H_0$ as a binary random variable with $\mathcal{D}^{val}$ viewed as fixed evidence. Observing $\mathcal{D}^{val}$ updates prior beliefs about $H_0$ into a posterior distribution via Bayes' theorem:
\begin{equation}
    \Pr(H_0 \mid \mathcal{D}^{val}) \propto \Pr(\mathcal{D}^{val} \mid H_0) \Pr(H_0)
\end{equation}

We quantify the statistical likelihood of $H_0$ by integrating $\Pr(H_0 \mid \mathcal{D}^{val})$ over relevant intervals of the IMDP parameters. Thus, Bayesian validation yields a continuous measure of system-wide conformance rather than a binary one.

% We design a Bayesian validation approach that effectively reduces the problem to measuring the conformance of the parameters of the new system $M_{E'}$ with those of the previously verified abstraction $\mathcal{M}_E$. Recall from subsection~\ref{sub:abstraction} that we previously formulated an abstraction of the concrete system $M_E$ as an MDP. Leveraging the same abstraction framework -- including the abstraction function $\psi$ -- we similarly assume that $M_{E'}$ admits an over-approximation as an MDP sharing the structural state representation with $M_E$.

% \paragraph{Belief over model parameters.} We adopt a Bayesian approach to model the uncertain transition probabilities of system $M_{E'}$.
% Formally, we encode our uncertainty as a \emph{belief} over model parameters\footnote{Note that this differs from our (frequentist) perspective in the previous two subsections, with model parameters as constants and trajectories as random variables.}. Starting with a prior distribution $\mathbb{M}^{pr}$ reflecting initial beliefs about the parameterization of $M_{E'}$, we incorporate observed evidence from the validation dataset $\mathcal{D}^{val}$ to update our beliefs, yielding a posterior distribution $\mathbb{M}^{po}$ over the $M_{E'}$ parameters

Let $\mathbb{M}^{pr}$ denote the prior distribution over the parameters of $M_E$. Since estimating discrete future states probabilistically is a multinomial process, we select $\mathbb{M}^{pr}$ as the Dirichlet distribution, the conjugate prior of the multinomial distribution. Observed evidence from the validation dataset $\mathcal{D}^{val}$ is incorporated by directly updating the parameters of our Dirichlet prior, resulting in a posterior distribution $\mathbb{M}^{po}$ over the parameters of $M_{E'}$. For each state, this Dirichlet distribution is parameterized by concentration parameters $\bm{\alpha} = [\alpha_1, \alpha_2, \dots, \alpha_n]$, where each $\alpha_i$ represents the prior pseudo-count (or relative occurrence) associated with the corresponding discrete state estimate $\hat{s}^{\psi}_i$. These concentration parameters are initialized to $\bm{\alpha} = \bm{1}$ (the uniform prior).

% Let $\mathbb{M}^{pr}$ denote the prior belief about the parameters of $M_E$. Since probabilistic estimation of the discrete future states is a multinomial process, we choose $\mathbb{M}^{pr}$ to be the multinomial conjugate prior Dirichlet distribution. We incorporate observed evidence from the validation dataset $\mathcal{D}^{val}$ to update our beliefs by directly updating the parameters of our Dirichlet prior, yielding a posterior distribution $\mathbb{M}^{po}$ over the $M_{E'}$ parameters. For each state, the Dirichlet distribution over state estimates is parameterized by density concentration parameters $\bm{\alpha} = [\alpha_1, \alpha_2, \dots, \alpha_n]$, where each $\alpha_i$ captures the prior pseudo-count (or relative occurrence) associated with the $n$ possible discrete state estimates $\hat{s}^{\psi}_i$. These parameters are uniformly initialized with an unbiased and uninformative prior $\bm{\alpha} = \bm{1}$.

We update our Dirichlet prior $\mathrm{Dir}(\bm{\alpha}_0)$ with the observed estimate counts from the validation set $\mathcal{D}^{val}$ for each discrete state $s^\psi$. To this end, we first define the count of estimates witnessed for each state:
$$n_i = \left|\{\,j : \hat{s}_j^\psi = i,\ s_j^\psi = s^\psi \}\right| \quad\text{for }i=1,\dots,K.$$

% The prior is modeled using the Dirichlet distribution -- the conjugate prior of the multinomial likelihood, which generalizes the Beta distribution to multiple binomial processes (the selection of a discrete state estimates in our case).

The posterior Dirichlet parameters for each $s^\psi$ are computed based on the counts $n_i$ and the prior parameters $\boldsymbol{\alpha}_0 = (\alpha_{0,1}, \dots, \alpha_{0,K})$ as $\alpha_i =\alpha_{0,i} \;+\; n_i$ for all $i \in \{1,\dots,K\}$, or simply $\bm{\alpha} = \bm{\alpha}_0 + \boldsymbol{n}$.

\paragraph{Validity confidence.} The next step is to quantify the conformance between this posterior distribution of the $M_{E'}$ parameters and the previously derived abstraction IMDP $\mathcal{M}_E$. We do this on a state-wise basis: a per‐state conformance confidence from posterior $\mathrm{Dir}\left(\boldsymbol\alpha\right)$ is defined as:  
\begin{equation}
1- \gamma_{s^\psi} =  \int_{\Delta [\cdot , s^\psi]}\mathrm{Dir}\bigl(\mathbf p;\,\boldsymbol\alpha\bigr)d\mathbf{p},
\end{equation}
\looseness=-1
where $\mathbf p \sim \mathrm{Dir}\left(\boldsymbol\alpha\right)$, and $\Delta[\cdot, s^\psi]$ are the intervals over possible state estimates $\hat{s}^\psi$.

In practice, due to the absence of an analytical solution to the above integral, % region $\Delta$ is often high-dimensional and irregularly shaped, 
we approximate $1-\gamma$ by drawing $N$ independent samples $\{\mathbf{p}^{(i)}\}_{i=1}^N$ from $\mathrm{Dir}(\bm{\alpha})$ and compute it as follows for each abstract state $s^\psi$: %, which converges without suffering the exponential blow-up of grid-based quadrature in numeric integration.
\begin{equation}
1-\gamma_{s^\psi} = \frac{1}{N} \sum_{i=1}^N \mathbf{1} \left(\mathbf{p}^{(i)} \in \Delta[\cdot, s^{\psi}] \right)
\end{equation}

We put these ideas of state-wise conformance checking together into a validation pipeline seen in Algorithm~\ref{alg:val}, which outputs a set of confidence values $\Gamma = \texttt{Validate} (\Delta, \psi, \mathcal{D}^{val})$ over all discrete states.

\vspace{-4mm}
\begin{algorithm}[htb]
\caption{\texttt{Validate}: compute the conformance of a system to an abstraction}
\begin{algorithmic}[1]
\Require Abstraction parameters $\Delta$, abstraction function $\psi$, validation data $\mathcal{D}^{val}$, Dirichlet prior $\bm{\alpha}_0$
\State Initialize $\Gamma \gets \{\}$ \Comment{Confidence levels over states}
\State $\bm{s}^{\psi} \gets [\psi(\bm{s}_i) \mid (\bm{s}_i, \bm{\hat{s}}_i) \in \mathcal{D}^{train}]$ \Comment{Bin the state data} \label{line:val0}
\State $\bm{\hat{s}^{\psi}} \gets [\psi(\bm{\hat{s}}_i) \mid (\bm{s}_i, \bm{\hat{s}}_i) \in \mathcal{D}^{train}]$ \label{line:val1}
\ForAll{unique $s^{\psi} \in \bm{s}^{\psi}$}
\State Let $I \gets \{ i \mid \bm{s}^{\psi}[i] = s^{\psi} \}$ \Comment{Vector indices where $s^{\psi}$ exists} \label{line:val2}
\State Let $\hat{\bm{s}}^{\psi}_{s^{\psi}} \gets [ \bm{\hat{s}^{\psi}}[i] \mid i \in I_{s^{\psi}} ]$ \Comment{Subset of $\bm{\hat{s}^{\psi}}$ indexed by $I$}
\State Let $\bm{n}  \gets   \left[ \left| \{i \mid  \hat{\bm{s}}^{\psi}_{s^{\psi}}[i] = s  \}  \right|   \mid  s \in \mathrm{unique}(\hat{\bm{s}}^{\psi}_{s^{\psi}})  \right]$ \label{line:val3}
\State Update $\bm{\alpha} \gets \bm{\alpha}_0 + \bm{n}$\label{line:val4}

\State\hspace{-5mm}\textbf{Monte Carlo integration:}
\State Draw $\{\mathbf{p}^{(i)}\}_{i=1}^N \overset{\mathrm{iid}}{\sim}\operatorname{Dir}(\bm{\alpha})$ \label{line:val5}
\State $ \gamma_{s^{\psi}} \gets 1- (1/N)\sum_{i=1}^N \mathbf{1}\bigl(\mathbf{p}^{(i)} \in \Delta[\cdot, s^{\psi}]\bigr)$
\State $\Gamma \gets \Gamma \cup 1-\gamma_{s^{\psi}}$ \label{line:val6}
\EndFor
\end{algorithmic}
\label{alg:val}
\end{algorithm}
\vspace{-5mm}

\looseness=-1
There are several ways of aggregating the state-wide confidences $\Gamma$ into the system-level confidence $\gamma$. One conservative system‐level metric of conformance is the worst‐case confidence $\min_{s^\psi\in S^\psi} \Gamma$. However, in practice, some bins $s^\psi$ may have only a few data points from $\mathcal D^{val}$, yielding low counts $n_i$ and thus an excessively low posterior confidence, even if most of the state space is well‐covered by $\mathcal D^{val}$. This overly low confidence has high uncertainty, can lead to numeric instability, and generally does not lead to a high validation precision due to the sensitivity to data availability, making it a poor indicator of system-wide conformance.
% 1-\gamma^{s^\psi}
% While a threshold on data availability can combat this,
To overcome the above limitation, we propose aggregating state-wide confidences using the \textit{median} among all states, which is robust to outliers ($ \gamma = \operatorname{median}(\Gamma)$). Furthermore, the median reflects the confidence of at least half of the states. As we show in Section~\ref{sec:exp}, our selection of the median metric allows for easy discrimination between the validity of $\mathcal{M}_E$ on the old ($M_E$) and new ($M_{E'}$) systems. 
% \begin{equation}
%     \gamma = \operatorname{median}%_{s^\psi\in S^\psi} 
%     (\Gamma) %[1-\gamma^{s^\psi}]
% \end{equation}

Finally, this section wraps up with a theorem that states a theoretical guarantee for our most conservative method. It asserts the lower confidence bound on the validity of a probabilistic safety claim under a novel environment $E'$:

\begin{theorem} [Out-of-distribution validity and safety]
\label{thm:safety-E}
    Let $M_{E'}$ be a concrete system subject to stochasticity from the latent environment distribution $E'$, from which originates training data $\mathcal{D}^{val}$ consisting of random trajectories $\tau$ drawn i.i.d. from $\sem{M_{E'}}$; let $\psi$ be an abstraction function; let $\mathcal{M}_E$ be an IMDP abstraction of the same system under $E$.
    
    If the following is true: 
    \begin{itemize}
        \item $\mathcal{M}_E$ is a statistically $\gamma$-valid abstraction of $M_{E'}$ with confidence $\gamma$  per Def.~\ref{def:prob-valid}. 
        \item $\mathcal{M}_E$ is safe for LTL safety property $\varphi$ with probability $1-\beta$ per Def.~\ref{def:prob-safe}. 
    \end{itemize}
    Then: 
    $$\Pr\!_{\mathcal{D}^{val}} \left[\Pr\!_{\tau \sim \sem{M_E}} \left( \tau \vDash \varphi \right) \ge 1-\beta \right] \ge 1-\gamma$$
\end{theorem}

The proof can be found in Appendix~\ref{app:proof-ood}.

\paragraph{Running example.} Consider the mountain car system introduced earlier, with state $s=(x,v) \in [-1.2,0.6]\times[-0.07,0.07]$. Suppose we have an abstraction $\mathcal{M}_E$, and aim to measure the extent to which it describes the behavior of the novel system $M_{E'}$.

We begin by gathering a new dataset $\mathcal{D}^{val}$ of trajectories from $M_{E'}$. Using the same abstraction function $\psi$ as before, each state is abstracted to discrete state tiles (e.g., a concrete state $s_k=(-0.3,0.06)$ maps to $\psi(s_k)=(-6,12)$), and estimation differences $x - \hat{x}$ map to discrete intervals ${-5,\dots,5}$.

Algorithm~\ref{alg:val} proceeds by binning validation data into abstract states (Lines~\ref{line:val0}--\ref{line:val1}), then iterating through each state tile, counting the occurrences of each abstract state seen at the current tile (Lines~\ref{line:val2}--\ref{line:val3}), and updating the prior parameters $\bm{\alpha}_0$ (Line~\ref{line:val4}). For example, suppose at $s^\psi = (-6,12)$, we observe the following state estimate relative occurrences: $\bm{n} = \{0, 3, 7, 11, 7, 0, \dots, 0\}$. Since $\bm{\alpha}_0$ is uniformly initialized, the posterior parameters become $\bm{\alpha} = \{1, 4, 8, 12, 8, 1, \dots, 1\}$.

We then approximate the integral for conformance confidence via Monte Carlo sampling (Lines~\ref{line:val5}--\ref{line:val6}). Drawing samples from the posterior Dirichlet distribution, we check whether each sampled probability vector falls within the previously computed confidence intervals $\Delta[\ \cdot \ ,(-6,12)]$. The fraction of samples within these intervals estimates our confidence $1-\gamma_{(-6,12)}$.

Repeating this process for each abstract state, we obtain confidence values $\Gamma$ across the entire discrete state space. Finally, we aggregate these confidences using the median to yield a robust system-level conformance metric $\gamma$, effectively validating the IMDP abstraction under environment $E'$.

\section{Experimental Evaluation}
\label{sec:exp}

Our experimental evaluation has three goals: (1) to demonstrate formal verification of safety for a vision-based autonomous system through an IMDP abstraction, (2) to check whether the validation framework effectively discriminates between systems deployed to in-distribution (ID) and out-of-distribution (OOD) scenarios, and (3) to characterize the tradeoffs among the confidence parameters involved in statistical abstraction and system-wide safety guarantees (namely, abstraction confidence $\alpha$, safety chance $1-\beta$, and conformance confidence $\gamma$). All experimental computation was performed on a modern laptop equipped with a 14th Gen Intel Core i9-14900HX 24-core processor, 64 GB DDR5 RAM, and an RTX 4090 GPU with 16 GB VRAM.

\subsection{Synthetic Goal-Reaching System}

\paragraph{Environment Description.} We consider a synthetic autonomous system $M_E$ where an agent must reach a waypoint located at coordinates $(10,10)$, starting from $(0.0, 0.0)$, within a two-dimensional state space $S = [0,12]\times[0,12]$. The agent uses a proportional controller with gain $0.5$, producing steps toward the waypoint with a maximum allowable magnitude of $0.7$. An execution is considered successful if the agent reaches within a radius of $2.0$ units of the waypoint in fewer than $100$ discrete time steps, without colliding with the barricades at $x=12$ and $y=12$. A collision or timeout constitutes an unsuccessful execution.

The agent has perfect localization; however, the waypoint's position is subject to uncertain estimation modeled by additive, unbiased Gaussian sensor noise with variance bounded within the interval $[\sigma^2_{\text{min}}, \sigma^2_{\text{max}}]$, which varies across the experiments. The magnitude of this noise decays as the agent approaches the waypoint, with the variance given by:
\begin{equation*}
\sigma^2(d) = \frac{(d - d_{\text{min}})(\sigma^2_{\text{max}} - \sigma^2_{\text{min}})}{d_{\text{max}} - d_{\text{min}}} + \sigma^2_{\text{min}},
\end{equation*}
where $d$ is the agent's Euclidean distance from the waypoint, bounded by $d_{\text{min}} = 2.0$ and $d_{\text{max}} = 14.14$ (the distance to the waypoint from the initial position).

\paragraph{Abstraction.} To construct an $\alpha$-sound abstraction of this goal-reaching system, we define an abstraction function $\psi$ using uniform discretization with bins of size $0.5$ units, partitioning the continuous 2-dimensional state space $[0,12]\times[0,12]$ into $576$ discrete tiles. We instantiate a training environment $E$ with fixed noise variance parameters $\sigma^2_{\text{min}} = 0.5$ and $\sigma^2_{\text{max}} = 0.5$, placing the agent at each discrete state tile and collecting $200$ state estimates per tile to form the dataset $\mathcal{D}^{\text{train}}$. Using this dataset, we compute the IMDP parameters via $\Delta = \texttt{ConfInt}\left(\mathcal{D}^{\text{train}}, \psi, \alpha \right)$, and subsequently construct the IMDP through $\texttt{IMDP} = \texttt{DynStruct}(\pi, f, S, \hat{S}, \psi)$, where the functions $\pi, f$ are inferred from closed-loop executions of $M_E$. Finally, we programmatically translate the IMDP into \texttt{PRISM} model syntax, resulting in approximately $16,000$ transitions. For this chosen level of state-space granularity, the model construction procedure requires roughly $120$ seconds of computation time on our laptop.

To confirm that the data-driven abstractions are a sound and over-approximate representation of the underlying system, we examined both the percentage of real trajectories covered by the abstraction, as well as the Jaccard similarity between the set of real and abstract transitions exiting a node, obtaining $100\%$ coverage and $\approx0.2$ similarity (depending on the value of $\alpha$).

\paragraph{Verification.}
We probabilistically model check the following PCTL property:
\begin{equation*}
\Pr\!^{\text{min}}_{\ge ?}\left[ (x < 12 \wedge y < 12) \;\mathsf{U}_{[0,100]}\; \sqrt{(x-10)^2 + (y-10)^2} \leq 2 \right]
\end{equation*}
which ensures that the agent never collides with a barricade and reaches within $2.0$ of the waypoint within $100$ time steps. The model was successfully verified at $10$ notable positions, and the verification results are summarized in Table~\ref{tab:synth-safety}.

\begin{table}
\vspace{-4mm}
\centering
\caption{Verified safety for the goal-reaching system from initial positions w/$\alpha = 0.05$}
\begin{tabular}{|>{\centering}p{1cm}|>{\centering}p{1cm}| p{1.1cm}|}
\hline
$x_0$ & $y_0$ & ~$1-\beta$ \\
\hline
0.0  & 0.0  & ~0.6909 \\
0.0  & 2.5  & ~0.7201 \\
2.5  & 0.0  & ~0.7202 \\
2.5  & 2.5  & ~0.7466 \\
0.0  & 5.0  & ~0.7465 \\
5.0  & 0.0  & ~0.7465 \\
5.0  & 5.0  & ~0.7934 \\
0.0  & 7.5  & ~0.7689 \\
7.5  & 0.0  & ~0.7855 \\
7.5  & 7.5  & ~0.8430 \\
\hline
\end{tabular}
\label{tab:synth-safety}
\vspace{-6mm}
\end{table}

\paragraph{Verification and Validation Tradeoff.}
We analyze the tradeoff between statistical abstraction confidence $1-\alpha$ and the resulting safety guarantees by varying the binomial confidence parameter $\alpha$ used in the construction of $\Delta$. We plot these safety probabilities against $\alpha$ in Figure~\ref{fig:synth-tradeoff}.

In addition to the model checking results, we also leverage the abstraction parameters $\Delta$ to measure how well the abstraction conforms to the in-distribution system through Algorithm~\ref{alg:val}. We instantiate another in-distribution environment with the same parameters as the training system ($\sigma^2 \in [0.5, 4.0]$), then collect $1000$ independent trajectories of state/estimate pairs, which yields the dataset $\mathcal{D}^{val}$. From the output, we obtain a representative confidence level $1-\gamma  = \mathsf{median} \left[  \texttt{Validate}(\Delta, \psi, \mathcal{D}^{val}) \right]$, and plot this against $\alpha$ in Figure~\ref{fig:synth-tradeoff}.

\begin{figure}[htb]
\vspace{-4mm}
\centering
\includegraphics[scale=0.5]{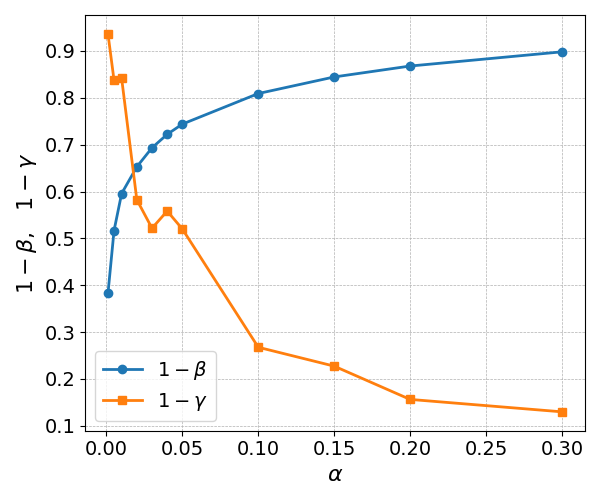}
\vspace{-4mm}
\caption{Safety chance and conformance confidence are tradeoffs as functions of $\alpha$ for the goal-reaching system. A smaller $\alpha$ corresponds to higher binomial confidence, leading to looser IMDP intervals, thus providing the model checker with greater flexibility to select worst-case distributions, resulting in lower safety chances.}
\label{fig:synth-tradeoff}
\vspace{-4mm}
\end{figure}

\paragraph{Discriminating ID and OOD systems.} To demonstrate the discriminative ability of our validation, we instantiate the agent in multiple domain-shifted environments, collect $1000$ independent trajectories from each, and compute conformance confidences used as classification scores for the ID vs. OOD decision. 

A domain shift is induced by biasing the distribution of state estimate noise with a scalar offset. To label which environments are truly outside of our IMDP, we compare their data-driven $99\%$-confidence interval with the model-checked interval from the IMDP. If these intervals do not overlap, the environment is deemed to be OOD for the IMDP. We also collect four datasets from the non-shifted environment, which are considered ID for the IMDP.

Our validation results are shown in Table~\ref{tab:safety-vs-distance}. Clearly, many reasonable thresholds (e.g., 0.5) on the $\gamma$ values would lead to a perfect discrimination between the ID and OOD cases. Note that the environments with smaller shifts (between 0 and 2.45) could not be conclusively labeled ID or OOD, so they were excluded.

\paragraph{Effects of discretization granularity.} We performed a small-scale ablation to assess how discretization granularity (tile size) affects verification and validation outcomes, holding the abstraction confidence at $\alpha = 0.05$, and checking the above property. As the state space was partitioned into finer tiles, the safety chance increased while the conformance confidence decreased. For example, with $1.0 \times 1.0$ tiles, the safety chance was $0.036$ (model checking time: $7.8s$) and conformance confidence was $0.464$. Refining to $0.8 \times 0.8$ tiles, safety chance increased to $0.067$ ($13.8s$) and conformance confidence was $0.450$. At $0.75 \times 0.75$, safety chance rose to $0.695$ ($15.4$) with conformance confidence $0.314$. Using $0.5 \times 0.5$ tiles, safety chance increased to $0.9998$ ($15.4s$), while conformance confidence was $0.406$.

\begin{table}[h]
\centering
\caption{Validation confidences and empirical safety for four in-distribution and ten out-of-distribution environments. The ID safety chance from \texttt{PRISM} is $[0.5837, 1.0000]$.}
\label{tab:safety-vs-distance}
{\setlength{\tabcolsep}{8pt}%
\begin{tabular}{ccccc}
\toprule
Case & Distance & Success rate & 99\% success CI    & $1-\gamma$  %& $\mathsf{mean}(1-\gamma)$ 
\\
\midrule
ID & 0.00                  & 1.000        & [0.9947, 1.0000]  & 0.7782  %&  0.6864       
\\
ID & 0.00                  & 1.000        & [0.9947, 1.0000]  & 0.7593 % &  0.6645        
\\
ID & 0.00                  & 1.000        & [0.9947, 1.0000]  & 0.7093 % &  0.5976       
\\
ID & 0.00                  & 1.000        & [0.9947, 1.0000]  & 0.6530 % &  0.6482      
\\
OOD & 2.45                  & 0.505        & [0.4638, 0.5461]  & 0.0000  %&  2.0e-5         
\\
OOD & 2.47                  & 0.454        & [0.4132, 0.4952]  & 0.0000 % &  0.0000     
\\
OOD & 2.50                  & 0.405        & [0.3651, 0.4459]  & 0.0000  %&  1.1e-6      
\\
OOD & 2.52                  & 0.391        & [0.3514, 0.4317]  & 0.0000 % &  3.2e-6     
\\
OOD & 2.55                  & 0.289        & [0.2527, 0.3274]  & 0.0000  %&  0.0000       
\\
OOD & 2.60                  & 0.263        & [0.2278, 0.3004]  & 0.0000 % &  0.0000        
\\
OOD & 2.70                  & 0.176        & [0.1461, 0.2091]  & 0.0000  %&  0.0000     
\\
OOD & 2.80                  & 0.092        & [0.0700, 0.1180]  & 0.0000  %&  0.0000       
\\
OOD & 2.90                  & 0.061        & [0.0431, 0.0832]  & 0.0000 % &  0.0000       
\\
OOD & 3.00                  & 0.025        & [0.0141, 0.0407]  & 0.0000  %&  0.0000       
\\
\bottomrule
\end{tabular}}
% \vspace{-5mm}
\end{table}

\subsection{Vision-based Mountain Car}

\looseness=-1
\paragraph{Environment Description.} We evaluate our methodology using the Mountain Car benchmark used as a running example in Section~\ref{sec:app}. Here we briefly augment the description of this system with details relevant to the experiment.

For our experiments, we trained a lightweight Deep Q-Network (DQN) policy over $5000$ episodes with $\epsilon_0=1.0$ and a decay rate of $0.9983$. During closed-loop execution with non-linear dynamics, the policy is guided by intermediate position estimates generated by a pre-trained vision-based estimator~\cite{waite_state-dependent_2025}, which predicts the car’s position directly from a noisy image.

To induce noise in vision-based state estimation, the images are perturbed with Gaussian noise before being input to the vision model. We modeled the "in-distribution" environment by applying unbiased Gaussian noise with $\sigma =0.1$ to the grayscale MountainCar images. Figure~\ref{fig:noisy-frames} shows example images perturbed with noise with parameters $\sigma \in \{0.00, 0.10, 0.25, 0.50\}$.

\begin{figure}[H]
\vspace{-4mm}
\centering
\includegraphics[scale=0.7]{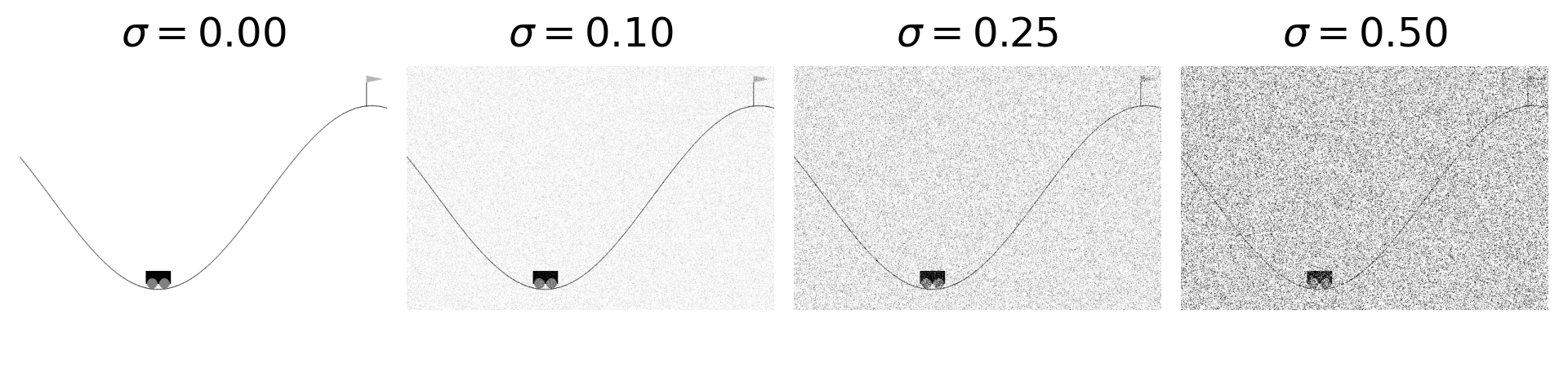}
\vspace{-7mm}
\caption{Examples of noisy visual inputs from the MountainCar environment. The true car position is fixed at $-0.46$. The position estimation errors are $0.01$, $0.25$, $0.46$, and $0.58$ from left to right.}
\label{fig:noisy-frames}
\vspace{-4mm}
\end{figure}

\paragraph{Abstraction.} To construct an $\alpha$-sound abstraction of the MountainCar, we define an abstraction function $\psi$ using uniform discretization with bins of size $0.05$ for $x$ and $0.005$ for $v$, partitioning the continuous 2-dimensional state space $[-1.2, 0.6] \times [-0.07,0.07]$ into $1008$ discrete tiles. We instantiate a training environment $E$ with fixed noise standard deviation $\sigma=0.1$. To collect the dataset $\mathcal{D}^{\text{train}}$, we place the mountain car at each of the discrete state tiles and collect $100$ observations and subsequent position estimates. Using this dataset, we compute the IMDP parameters via $\Delta = \texttt{ConfInt}\left(\mathcal{D}^{\text{train}}, \psi, \alpha \right)$, and subsequently construct the IMDP through $\texttt{IMDP} = \texttt{DynStruct}(\pi, f, S, \hat{S}, \psi)$. In this experiment, $\pi$ is a DQN policy, and the mountain car dynamics $f$ are inferred by propagating the gym environment over one time step for particular state/action pairs. Finally, we programmatically translate the IMDP into \texttt{PRISM} model syntax, resulting in approximately $23,000$ transitions. For this chosen level of state-space granularity, the model construction procedure requires roughly $300$ seconds of computation time on a modern laptop. Similarly to the first case study, this abstraction is conservative: it has $100\%$ coverage of concrete behaviors and $\approx0.2$ Jaccard similarity between the set of real and abstract transitions exiting a node.

\paragraph{Verification.} 
We probabilistically model check the following PCTL property:
$$\Pr\!^{\text{min}}_{\ge ?} \left[  F_{[0,200]} \left(x \ge 0.45\right) \right]$$
which ensures that the car reaches the goal within 200 discrete time steps. The model was successfully verified at states such as $(0.2, 0.07)$, which yields a $1-\beta$ safety chance of $0.9991$, and $(0.3, 0.06)$, which yields the same safety chance.

\vspace{-2mm}

\paragraph{Verification and Validation Tradeoff.} 
We again analyze the tradeoff between the lower bound on safety of the IMDP, the conformance of a novel system to this IMDP, and the confidence level used to construct the binomial CIs for the IMDP from real data. We analyze the tradeoff between statistical abstraction confidence $1-\alpha$ and the resulting safety guarantees by varying the binomial confidence parameter $\alpha$ exponentially between $0.001$ and $0.3$ to construct $\Delta$.

\looseness=-1
We then collect another dataset $\mathcal{D}^{val}$ of state/estimate trajectories over $50$ episodes in the environment with $\sigma=0.1$. Then, for each $\alpha$, we tabulate the lower bound of safety of the abstraction $1-\beta$ and the median representative conformance confidence of the abstraction to the data $\mathcal{D}^{val}$ computed as $1-\gamma  = \mathsf{median} \left[  \texttt{Validate}(\Delta, \psi, \mathcal{D}^{val}) \right]$, producing the results seen in Figure~\ref{fig:mc-tradeoff}. They confirm the observation that model soundness attenuates the safety-validity tradeoff. 

\paragraph{Limitations.}
The most obvious limitation of our approach is its scalability: our abstraction method produces sizable $\texttt{PRISM}$ models that consume time and computational resources when model checking. This makes it challenging to verify systems with non-linear dynamics and controls (such as the MountainCar) because they require carefully balancing non-determinism between overly conservative and overly large abstractions. Nonetheless, our theoretical approach is not inherently limited to low dimensions. This issue can be mitigated with more sophisticated abstraction in the future, like adaptive tiling. Secondly, measuring the conformance of the IMDP parameters at each discrete state is sample-inefficient. We anticipate that continuous statespace-wide distribution models, such as Gaussian processes, will be able to overcome this limitation.

\begin{figure}[H]
\vspace{-4mm}
\centering
\includegraphics[scale=0.5]{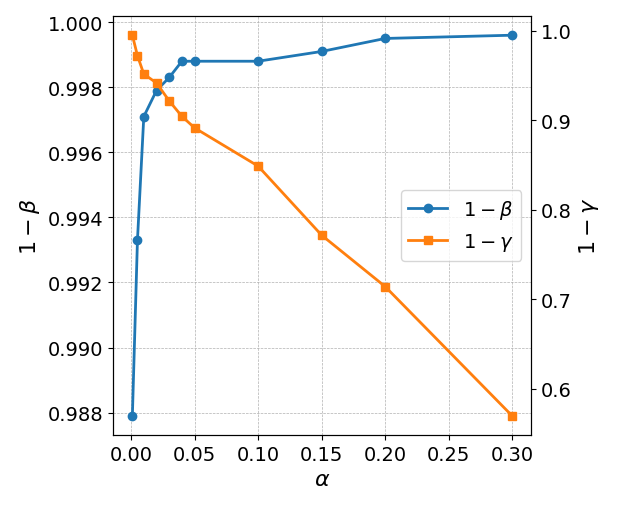}
\vspace{-4mm}
\caption{Safety chance and conformance confidence are tradeoffs as functions of $\alpha$ for the mountain car, analogous to Figure~\ref{fig:synth-tradeoff}.
}
\label{fig:mc-tradeoff}
\vspace{-4mm}
\end{figure}

\section{Conclusion}
\label{sec:conc}

\vspace{-1mm}

This paper contributes a unified verification \& validation methodology for assuring vision-based autonomous systems. The methodology contains three steps: (1) statistical \textbf{soundness} of the vision-based system abstraction, (2) probabilistic \textbf{verification} of the safety of the vision-based system, and (3) statistical \textbf{validation} of the abstraction in a new operational environment. Our experiments have shown that our abstractions are sound: the IMDP trajectories robustly over-approximate those of the true system. In both case studies, we achieved $100\%$ coverage of the real system trajectories with the IMDP, and a Jaccard similarity of approximately $0.2$ (no. of transitions shared between the real and abstract systems divided by the total no. of transitions). For verification, we indeed produce genuine safety chance lower bounds in \texttt{PRISM} model checking. For validation, system-wide conformance checking via Bayesian inference effectively discriminates between in-distribution and shifted environments. Finally, end-to-end guarantees are established via Theorems~\ref{thm:id-safety} and~\ref{thm:safety-E} and their proofs.

This paper opens several fruitful directions for future research. First, our validation procedure can be targeted towards the states that carry the most probability mass in the verification. Second, one can apply more sophisticated uncertainty models to create abstractions, such as Gaussian processes. Third, there are promising models of photorealistic neural scene representation, such as Neural Radiance Fields (NeRF) \cite{10.1145/3503250} and Gaussian splats \cite{kerbl3Dgaussians}, which enables photorealistic simulation of vision-based robotics applications like zero-shot sim-to-real transfer for aerial navigation \cite{Miao:IROS2025, miao2025falconwingopensourceplatformultralight}, self-driving \cite{xie2024vid2sim}, and manipulation tasks \cite{robosplat}; however, the soundness of those neural representation abstraction remains an open question. Lastly, state abstraction can be improved with more sophisticated procedures, such as adaptive tiling or polynomial representations. Furthermore, we plan to apply our approach to realistic physical systems, such as drone racing \cite{Miao:IROS2025} and vision-based fixed-wing landing \cite{miao2025falconwingopensourceplatformultralight}.

\bibliographystyle{splncs04}
\bibliography{bib/bib,bib/ivan-autogen}

\clearpage
\appendix 
\section{Appendix}

\setcounter{theorem}{0}
\subsection{Proof of Theorem 1}\label{app:proof-id}
\begin{theorem} [In-distribution soundness and safety]
%\label{thm:id-safety}
Let $M_E$ be a concrete system, living in the state/estimation spaces $S$ and $\hat{S}$, composed of deterministic components $g$, $h$, $\pi$, and $f$, subject to stochasticity from the latent environment distribution $E$, from which originates training data $\mathcal{D}^{train}$ consisting of random trajectories $\tau$ drawn i.i.d. from $\sem{M_E}$; let $\psi$ be an abstraction function; let $\mathcal{M}_E$ be an \emph{ $\texttt{IMDP} = \texttt{DynStruct}(\pi, f, S, \hat{S}, \psi)$} parameterized by \emph{$\Delta = \texttt{ConfInt}\left(\mathcal{D}^{train}, \psi, \alpha \right)$}, where $\alpha$ is a soundness confidence level.
    
    If it holds that:
    \begin{itemize}
        \item $\mathcal{M}_E$ is a statistically $\alpha$-sound abstraction of $M_E$ per Def.~\ref{def:stat-sound-absteaction}. 
        \item $\mathcal{M}_E$ is safe for LTL safety property $\varphi$ with probability $1-\beta$ per Def.~\ref{def:prob-safe}. 
    \end{itemize}

    then:
    $$
    \Pr\!_{\mathcal{D}^{train}} \left[\Pr\!_{\tau \sim \sem{M_E}} \left( \tau \vDash \varphi \right) \ge 1-\beta \right] \ge 1-\alpha
    $$
\end{theorem}

\begin{proof}
    For convenience, we define the events:
    \begin{itemize}
        \item $A \triangleq \sem{\psi(M_E)}\in\sem{\mathcal{M}_E}$ is the event that the distribution over concrete trajectories lives within the IMDP (\emph{statistical soundness} Def.~\ref{def:stat-sound-absteaction}).
        \item $B \triangleq \tau \vDash \varphi \mid \tau \sim \sem{\mathcal{M}_E}$ is the event that an execution of the IMDP $\mathcal{M}_E$ satisfies the LTL safety property $\varphi$ (\emph{probabilistic safety} Def.~\ref{def:prob-safe}).
        \item $C \triangleq \tau \vDash \varphi \mid \tau \sim \sem{M_E}$ is the event that an execution $M_E$ satisfies the LTL safety property $\varphi$.
    \end{itemize}

    \noindent
    Suppose $A$ occurs. By over-approximation, $\tau \sim \sem{M_E}$ admits an abstract representative $\psi(\tau) \sim \sem{\mathcal{M}_E}$. Moreover, if $\psi(\tau) \vDash \varphi$ then the original trajectory $\tau$ must satisfy $\varphi$ as well. Hence:
    $$A \implies (B \implies C)$$

    \noindent
    By Definition~\ref{def:prob-safe}, we know that $\Pr(B) \ge 1-\beta$. Then we know that $\tau$ is \emph{at least} as safe as $\psi(\tau)$:
    $$\Pr(C) = \Pr\!_{\tau \sim \sem{M_E}} \left( \tau \vDash \varphi \right) \ge \Pr\!_{\psi(\tau) \sim \sem{\mathcal{M}_E}}(\psi(\tau) \vDash \varphi) \ge 1-\beta$$

    \noindent
    Furthermore:
    $$\Pr \left [  \Pr(C) \ge 1-\beta   \mid A \right] = 1$$
    By the Law of Total Probability:
    $$\Pr \left [  \Pr(C) \ge 1-\beta \right] = \Pr \left [  \Pr(C) \ge 1-\beta   \mid A \right] \Pr[A] + \Pr \left [  \Pr(C) \ge 1-\beta   \mid A^c \right] \Pr[A^c] $$
    $$\Pr \left [  \Pr(C) \ge 1-\beta \right] \ge (1) (1 - \alpha) + (0)(\alpha) $$
    Which reduces to the desired PAC guarantee:
    $$\Pr \left [  \Pr(C) \ge 1-\beta \right] \ge 1 - \alpha$$
    
\end{proof}

\subsection{Proof of Theorem 2}\label{app:proof-ood}

\begin{theorem} [Out-of-distribution validity and safety]
%\label{thm:safety-E}
    Let $M_{E'}$ be a concrete system subject to stochasticity from the latent environment distribution $E'$, from which originates training data $\mathcal{D}^{val}$ consisting of random trajectories $\tau$ drawn i.i.d. from $\sem{M_{E'}}$; let $\psi$ be an abstraction function; let $\mathcal{M}_E$ be an IMDP abstraction of the same system under $E$.
    
    If the following is true: 
    \begin{itemize}
        \item $\mathcal{M}_E$ is a statistically $\gamma$-valid abstraction of $M_{E'}$ with confidence $\gamma$  per Def.~\ref{def:prob-valid}. 
        \item $\mathcal{M}_E$ is safe for LTL safety property $\varphi$ with probability $1-\beta$ per Def.~\ref{def:prob-safe}. 
    \end{itemize}
    Then: 
    $$\Pr\!_{\mathcal{D}^{val}} \left[\Pr\!_{\tau \sim \sem{M_E}} \left( \tau \vDash \varphi \right) \ge 1-\beta \right] \ge 1-\gamma$$
\end{theorem}

\begin{proof}

    For convenience, we define the following events:
    \begin{itemize}
        \item $A \triangleq \sem{M} \in \sem{\mathcal{M}_E} \mid M \sim \mathbb{M}^{po}$ is the event that a model in our posterior distribution over models exists within the IMDP (\emph{statistical validity} Def.~\ref{def:prob-valid}).
        \item $B \triangleq \tau \vDash \varphi \mid \tau \sim \sem{\mathcal{M}_E}$ is the event that an execution of the IMDP $\mathcal{M}_E$ satisfies the LTL safety property $\varphi$ (\emph{probabilistic safety} Def.~\ref{def:prob-safe}).
        \item $C \triangleq \tau \vDash \varphi \mid \tau \sim \sem{M_{E'}}$ is the event that an execution of $M_{E'}$ satisfies the LTL safety property $\varphi$.
    \end{itemize}
    \noindent
    Suppose $A$ occurs. By over-approximation, $\tau \sim \sem{M_{E'}}$ admits an abstract representative $\psi(\tau) \sim \sem{\mathcal{M}_E}$. Moreover, if $\psi(\tau) \vDash \varphi$ then the original trajectory $\tau$ must satisfy $\varphi$ as well. Hence:
    $$A \implies (B \implies C)$$

    \noindent
    By Definition~\ref{def:prob-safe}, we know that $\Pr(B) \ge 1-\beta$. Then we know that $\tau$ is \emph{at least} as safe as $\psi(\tau)$:
    $$\Pr(C) = \Pr\!_{\tau \sim \sem{M_E}} \left( \tau \vDash \varphi \right) \ge \Pr\!_{\psi(\tau) \sim \sem{\mathcal{M}_E}}(\psi(\tau) \vDash \varphi) \ge 1-\beta$$

    \noindent
    Furthermore:
    $$\Pr \left [  \Pr(C) \ge 1-\beta   \mid A \right] = 1$$
    By the Law of Total Probability:
    $$\Pr \left [  \Pr(C) \ge 1-\beta \right] = \Pr \left [  \Pr(C) \ge 1-\beta   \mid A \right] \Pr[A] + \Pr \left [  \Pr(C) \ge 1-\beta   \mid A^c \right] \Pr[A^c] $$
    $$\Pr \left [  \Pr(C) \ge 1-\beta \right] \ge (1) (1 - \gamma) + (0)(\gamma) $$
    Which reduces to the desired PAC guarantee:
    $$\Pr \left [  \Pr(C) \ge 1-\beta \right] \ge 1 - \gamma$$

\end{proof}

\end{document}